\documentclass[11pt, journal, draftclsnofoot, onecolumn]{IEEEtran}

\usepackage{mathtools,enumitem,color,subfigure,bbm}
\usepackage{graphicx,tabularx,array}
\usepackage{hyperref}

\usepackage{amsthm,mathrsfs}
\usepackage{makecell}
\usepackage{multirow}
\usepackage[T1]{fontenc}
\usepackage{cite}
\usepackage{txfonts}
\usepackage{enumerate}

\IEEEoverridecommandlockouts

\DeclareRobustCommand{\rchi}{{\mathpalette\irchi\relax}}
\newcommand{\irchi}[2]{\raisebox{\depth}{$#1\chi$}} 

\global\long\def\RR{\mathbb{R}}

\global\long\def\11{\mathbbm{1}}

\global\long\def\+{\oplus}

\newcommand\pmm{\{-1,1\}}

\def\<{\langle}
\def\>{\rangle}

\DeclareMathOperator*{\argmin}{arg\,min}
\DeclareMathOperator*{\argmax}{arg\,max}

\usepackage{stackengine}
\def\deq{\mathrel{\ensurestackMath{\stackon[1pt]{=}{\scriptstyle\Delta}}}}

\newcommand\numberthis{\addtocounter{equation}{1}\tag{\theequation}}

\def\fS{{f}_{\mathcal{S}}}

\def\ps{\phi_\mathcal{S}}
\def\pS{\ps}

\newcommand{\pset}[1]{{\phi}_{#1}}

\def\L2{\mathcal{L}^2(\mathcal{X}, P_{X^d})}

\def\E_mu{\mathcal{E}_{\mu}(\epsilon')}

\definecolor{green}{rgb}{0.6627,0.8196,0.5568}
\definecolor{blue}{rgb}{0.6875,0.8750,0.8984}
\definecolor{purple}{rgb}{ 0.7647,    0.6078,    0.8824}

\newcommand{\discardpages}[1]{
  \xdef\discard@pages{#1}
  \AtBeginShipout{
    \renewcommand*{\do}[1]{
      \ifnum\value{page}=##1\relax%
        \AtBeginShipoutDiscard
        \gdef\do####1{}
      \fi%
    }%
    \expandafter\docsvlist\expandafter{\discard@pages}
  }%
}
\newif\ifkeeppage
\newcommand{\keeppages}[1]{
  \xdef\keep@pages{#1}
  \AtBeginShipout{
    \keeppagefalse%
    \renewcommand*{\do}[1]{
      \ifnum\value{page}=##1\relax%
        \keeppagetrue
        \gdef\do####1{}
      \fi%
    }%
    \expandafter\docsvlist\expandafter{\keep@pages}
    \ifkeeppage\else\AtBeginShipoutDiscard\fi
  }%
}

\newtheorem{Theorem}{Theorem}

\newtheorem{Lemma}{Lemma}

\newtheorem{Remark}{Remark}
\newtheorem{Definition}{Definition}

\usepackage{xcolor}
\definecolor{darkblue}{rgb}{0.1,0.1,0.8}
\definecolor{DarkGreen}{rgb}{0,0.6,0}
\definecolor{brickred}{rgb}{0.8, 0.25, 0.33}
\definecolor{britishracinggreen}{rgb}{0.0, 0.26, 0.15}
\definecolor{calpolypomonagreen}{rgb}{0.12, 0.3, 0.17}
\definecolor{ao(english)}{rgb}{0.0, 0.5, 0.0}
\definecolor{cadmiumgreen}{rgb}{0.0, 0.42, 0.24}
\definecolor{burgundy}{rgb}{0.5, 0.0, 0.13}

\ifCLASSINFOpdf

\begin{document}
%


\title{On Non-Interactive Source Simulation via Fourier Transform}

\author{

\IEEEauthorblockN{ Farhad Shirani$^\dagger$\thanks{This work was supported in part by NSF grants CCF-2241057 and CCF-2211423.}, Mohsen Heidari$^\ddagger$}
\IEEEauthorblockA{$^\dagger$Florida International University, $^\ddagger$ Indiana University, Bloomington
\\Email: fshirani@fiu.edu, mheidar@iu.edu}
}

\maketitle

 \begin{abstract}
The non-interactive source simulation (NISS)  scenario is considered. In this scenario, a pair of distributed agents, Alice and Bob, observe a distributed binary memoryless source $(X^d,Y^d)$ generated based on joint distribution $P_{X,Y}$. The agents wish to produce a pair of discrete random variables $(U_d,V_d)$ with joint distribution $P_{U_d,V_d}$, such that $P_{U_d,V_d}$ converges in total variation distance to a target distribution $Q_{U,V}$ as the input blocklength $d$ is taken to be asymptotically large.
Inner and outer bounds are obtained on the set of distributions $Q_{U,V}$ which can be produced given an input distribution $P_{X,Y}$. To this end, a bijective mapping from the set of distributions $Q_{U,V}$ to a union of star-convex sets is provided. By leveraging proof techniques from discrete Fourier analysis along with a novel randomized rounding technique, inner and outer bounds are derived for each of these star-convex sets, and by inverting the aforementioned bijective mapping, necessary and sufficient conditions on $Q_{U,V}$ and $P_{X,Y}$ are provided under which $Q_{U,V}$ can be produced from $P_{X,Y}$. The bounds are applicable in NISS scenarios where the output alphabets $\mathcal{U}$ and $\mathcal{V}$ have arbitrary finite size. In case of binary output alphabets, the outer-bound recovers the previously best-known outer-bound. 
\end{abstract}


%
\IEEEpeerreviewmaketitle
\vspace{-.06in}
\section{Introduction}

Non-interactive source simulation (NISS), shown in Figure \ref{fig:1}, refers to a scenario where two distributed agents, Alice and Bob,  each observe a sequence of random variables,  $X^d$ and $Y^d$, respectively, where $d\in \mathbb{N}$. The agents wish to produce $(U_d,V_d)\sim P_{U_d,V_d}$ via (possibly stochastic) functions $U_d=f_d(X^d)$ and $V_d=g_d(Y^d)$, such that the joint distribution $P_{U_d,V_d}$ is {close} to a target distribution $Q_{U,V}$ in total variation. The target distribution $Q_{U,V}$ is said to be feasible for an input distribution $P_{X,Y}$ if
\begin{equation}\label{eq:NISS}
\lim_{d\rightarrow \infty } \quad \inf_{\substack{(U_d,V_d)\sim P_{U_d,V_d}\\U_d-X^d-Y^d-V_d}}d_{TV}(P_{U_d,V_d},
Q_{U,V})=0.
\end{equation}

The NISS scenario arises naturally in various applications such as cryptography \cite{ahlswede1993common}, covert communications \cite{maurer1993secret}, game-theoretic coordination in adversarial networks \cite{anantharam2007common,kamath2016non}, quantum computing \cite{nielsen1999conditions},  and correlation distillation \cite{bogdanov2011extracting,li2020boolean}. 


Prior works have studied analytical characterization of the set of feasible distributions, and constructive algorithms for producing the associated variables under various assumptions on input distributions and output alphabets.   
Witsenhausen  studied the problem in \cite{witsenhausen1975sequences}, under the assumption that $(X^d, Y^d)$ are  independent and identically distributed (IID)  and that $(U,V)$ are jointly Gaussian and derived necessary and sufficient conditions on their correlation coefficient under which $Q_{U,V}$ is feasible for $P_{X,Y}$. An algorithm for constructing the corresponding mapping pair $(f(\cdot),g(\cdot))$ was provided, and it was shown that for jointly Gaussian target distributions, any correlation less than the Hirschfeld-Gebelein-R\'enyi maximal correlation coefficient \cite{hirschfeld1935connection,gebelein1941statistische,renyi1959measures} can be produced at the distributed terminals. There has been significant recent progress in the study of the  discrete version of the problem. In \cite{kamath2016non},
hypercontractivity techniques were used to provide necessary conditions on the distribution $Q_{U,V}$, i.e., impossibility results for NISS. In \cite{ghazi2016decidability}, NISS problems with binary outputs were shown to be decidable. That is, it was shown that given distributions $P_{X,Y}$ and $Q_{U,V}$, there exists a Turing machine that can decide in finite time whether $Q_{U,V}$ is feasible for a distribution $P_{X,Y}$. These decidability
results were further extended to general finite output alphabets in \cite{de2018non}. In \cite{yu2021non}, impossibility results were derived for the binary-input, binary-output case which improve upon the hypercontractivity-based bounds in \cite{kamath2016non}. These derivations rely on the Fourier expansion of the functions $(f(\cdot),g(\cdot))$ over the Boolean cube. 
The NISS problem is closely related to other two-agent problems including those considered by G\'acs and K\"orner \cite{gacs1973common}, and Wyner \cite{wyner1975common}. A comprehensive survey of relevant problems and their connections to NISS is given in \cite{yu2022common}. 

 \begin{figure}[t]
\centering 
\includegraphics[width=0.5\textwidth]{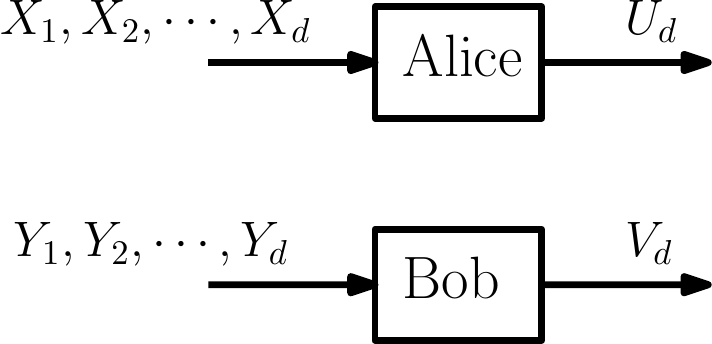}
\caption{The non-interactive source simulation problem. 
}
\vspace{-.25in}
\label{fig:1}
\end{figure}

In this paper, we consider the binary-input NISS problem, where the distributed binary source $(X^d,Y^d)$ is memoryless, and the output variables $(U,V)$ have general finite alphabets. We show that the set of feasible target distributions $Q_{U,V}$ can be bijectively mapped to a union of star convex sets, which can be evaluated by characterizing their extreme points and supporting hyperplanes. We find inner-bounds and outer-bounds for the set of feasible target distributions by evaluating these supporting hyperplanes. The outer bounds generalize those given in \cite{yu2021non} and the inner-bounds provide new achievability results for the NISS problem. The evaluation techniques rely on the Fourier expansion on the Boolean cube \cite{Wolf2008,o2014analysis} and its application in quantifying correlation among distributed functions \cite{mossel2004learning,courtade2014boolean,shirani2017correlation,shirani2019sub,
Heidari2019,yu2021non,ghazi2016decidability,Heidari2022}, and a novel randomized rounding technique which generalizes the one developed in \cite{ghazi2016decidability}. 
The main contributions of this work are summarized below:
\begin{itemize}[leftmargin=*]
    \item To provide a bijective mapping from the set of feasible distributions $Q_{U,V}$ to a union of star-convex sets. The mapping relies on discrete Fourier analysis techniques and a randomized rounding method described in Section \ref{sec:map}. 
    \item To derive inner-bounds and outer-bounds on the aforementioned star-convex sets. The derivation relies on methods for representing star-convex sets via their extreme points (Theorem \ref{th:1}).  
    \item To provide inner-bounds and outer-bounds on the set of feasible distributions $Q_{U,V}$  given an input distribution $P_{X,Y}$. The derivation relies on inverting the aforementioned bijective mapping, and Fourier analysis techniques (Theorem \ref{th:2}). 
\end{itemize}

{\em Notation:}
 The set $\{1,2,\cdots, d\}$ is represented by $[d]$. 
The  vector $(x_1,x_2,\hdots, x_d)$ is written as $x^d$. Sets are denoted by calligraphic letters such as $\mathcal{X}$. For the event $\mathcal{E}$, the variable $\mathbbm{1}(\mathcal{E})$ denotes the indicator of the event, and we define $\rchi(\mathcal{E})\triangleq 2\mathbbm{1}(\mathcal{E})-1$. The notation
$d_{TV}(P,Q)$ represents the variational distance between distributions $P$ and $Q$ defined on a shared alphabet $\mathcal{X}$, i.e., $d_{TV}(P,Q)\triangleq \sum_{x\in \mathcal{X}}|P(x)-Q(x)|$. For a given alphabet $\mathcal{X}$, the notation $\Delta_{\mathcal{X}}$ represents the probability simplex on $\mathcal{X}$. $Conv(\cdot)$ represents the convex-hull.

\section{Problem Formulation and Preliminaries}
In this section, we formally define the NISS problem and describe some of the techniques related to Fourier expansion on the Boolean cube which are used in our derivations. In its most general form, the NISS problem is defined as follows. 
\begin{Definition}[\textbf{Non-Interactive Source Simulation}] 
Let $P_{X^d,Y^d}, d\in \mathbb{N}$ be the sequence of probability measures corresponding to the jointly stationary and ergodic pair of stochastic processes  $(\mathsf{X},\mathsf{Y})= (X_i, Y_i), i\in \mathbb{N}$, and let $Q_{U,V}$ be a probability measure defined on finite alphabets $\mathcal{U}\times \mathcal{V}$. The distribution $Q_{U,V}$ is called feasible for the source $(X_i, Y_i), i\in \mathbb{N}$ if there exists a sequence of (possibly stochastic) functions $f_d:\mathcal{X}^d\to \mathcal{U}$ and $g_d:\mathcal{Y}^d\to \mathcal{V}$ such that $\lim_{d\to \infty} d_{TV}(P^{(d)},Q_{U,V})=0$, where $P^{(d)}$ is the joint distribution of $(f_d(X^d), g_d(Y^d))$. The sequence of pairs of functions $(f_d,g_d), d\in \mathbb{N}$ are called an associated sequence of functions of $Q_{U,V}$.
We denote the set of all feasible distributions by $\mathcal{P}(P_{\mathsf{X},\mathsf{Y}})$.
\label{def:1}
\end{Definition}
In this paper, we focus on the NISS scenarios where the stochastic processes $(\mathsf{X},\mathsf{Y})$ consist of IID pairs of random variables $(X_i,Y_i), i\in \mathbb{N}$ with distribution $P_{X,Y}$ and taking values from $\{-1,1\}$. 
Hence, we denote the set of all feasible distributions by $\mathcal{P}(P_{X,Y})$, i.e. as a function of the underlying single-letter distribution. The derivation techniques introduced in the sequel, which utilize discrete Fourier analysis, can potentially be applied to general ergodic sources using the Gram-Schmidt decomposition method introduced in \cite{heidari2021finding,Heidari2022}.

Without loss of generality, we assume that $\mathcal{U}=\{0,1,\cdots, |\mathcal{U}|-1\}$ and $\mathcal{V}=\{0,1,\cdots, |\mathcal{V}|-1\}$. Our derivations rely on Fourier analysis techniques which are briefly described in the sequel. For a complete discussion please refer to \cite{o2014analysis}.

\textbf{Boolean Fourier Expansion:}
\label{sec:Fourier}
Let $X^d$ be a vector of IID variables with alphabet $\mathcal{X}=\{-1,1\}$ and $P_X(1)=p\in (0,1)$.  Let $\mu_X=2p-1$ and $\sigma_X= 2\sqrt{p(1-p)}$ denote the mean and standard deviation, respectively. Consider the vector space $\mathcal{L}^d_X$ of functions $f_d:\{-1,1\}^d\to \mathbb{R}$ equipped with the inner-product operation $\<f_d(\cdot), g_d(\cdot)\>\triangleq \mathbb{E}(f_d(X^d)g_d(X^d))$. The Fourier expansion provides a decomposition of such functions via an orthonormal basis consisting of \textit{parities}. 
To elaborate, the parity associated with a subset $\mathcal{S}\subseteq [d]$ is defined as:
\begin{align*}
\pS(x^d)\deq \prod_{i\in \mathcal{S}}\frac{x_i-\mu_X}{\sigma_X},\qquad  x^d \in \{0,1\}^d.
\end{align*}
Since the source is IID, the parities are orthonormal. That is $\<\pS(X^d)~ ,\pset{\mathcal{T}}(X^d)\>=\mathbbm{1}(\mathcal{S}=\mathcal{T})$, where $\mathcal{S},\mathcal{T}\subseteq [d]$. The parities span $\mathcal{L}^d_{X}$.
That is, for any function $f_d\in \mathcal{L}_X^d$, we have:
\begin{align*}
 f_d(x^d)=\sum_{\mathcal{S}\subseteq [d]} \fS~\ps(x^d), \quad \text{for all}~ x^d \in \pmm^d,
 \end{align*}
  where $\fS\in \RR$ are called the {\it Fourier coefficients} of $f$ with respect to $P_X$, and are computed as  $\fS= \<f_d(X^d),\pS(X^d)\>$, for all $\mathcal{S}\subseteq [d]$. The Fourier expansion is closely related to the probability of disagreement between pairs of Boolean functions of sequences of random variables. In \cite{shirani2017correlation,Heidari2019} this relation was studied for binary-output pairs of functions, i.e. $|\mathcal{U}|=|\mathcal{V}|=2$. In particular, it was shown that given a pair of binary random variable variables $(X,Y)\sim P_{X,Y}$ and a pair of discrete-output functions $f_d,g_d\in \mathcal{L}_X^d\times \mathcal{L}_Y^d$, the following holds:
  \begin{align}\label{eq:1}
     & \mathbb{E}(f_d(X^d)g_d(Y^d))
      = \sum_{\mathcal{S}\subseteq [d]} f_{\mathcal{S}}g_{\mathcal{S}} \rho^{|\mathcal{S}|},\\
      \label{eq:1.5}
      &    P(f_d(X^d)\neq g_d(Y^d))= \frac{1-\mathbb{E}(f_d(X^d)g_d(Y^d))}{2}.
  \end{align}
  where $\rho\triangleq  \mathbb{E}\left(\frac{(X-\mu_X)(Y-\mu_Y)}{\sigma_X\sigma_Y}\right)$ is the Pearson correlation coefficient between $X$ and $Y$.  
  In this paper, we use an extension to arbitrary finite output alphabets as follows. Given functions $f_d:\{-1,1\}^d\to \mathcal{U}$ and $g_d:\{-1,1\}\to \mathcal{V}$, we define $f_{d,u}(\cdot)= \rchi(f_d(\cdot)=u)$ and $g_{d,v}(\cdot)= \rchi(g_d(\cdot)=v)$, where $u,v\in \mathcal{U}\times \mathcal{V}$ and for an event $\mathcal{E}$ we have defined $\rchi(\mathcal{E})\triangleq 2\mathbbm{1}(\mathcal{E})-1$. Then from \eqref{eq:1}-\eqref{eq:1.5} we have: 
  \begin{align}
      P(f_{d,u}(X^d)\neq g_{d,v}(Y^d))= \frac{1-\mathbb{E}(f_{d,u}(X^d)g_{d,v}(Y^d))}{2}.\label{eq:2}
  \end{align}
Let us recall that $U_d= f_{d,u}(X^d)$ and $V_d=g_{d,v}(Y^d)$. It can be noted that the distribution of $(U_d,V_d)$ is completely characterized by 
  $  P(f_{d,u}(X^d)\neq g_{d,v}(Y^d)), u,v\in \mathcal{U}\times \mathcal{V}$ and the marginals $P(f_{d,u}(X^d))=u, u\in \mathcal{U}$ and $P(g_{d,v}(Y^d))=v, v\in \mathcal{U}$. As a result, characterizing the set of feasible distributions is equivalent to characterizing the set of feasible values for $\mathbb{E}(f_{d,u}(X^d)g_{d,v}(Y^d)), u,v\in \mathcal{U}\times \mathcal{V}$.
%

\section{Mapping $\mathcal{P}(P_{X,Y})$ to a Star-Convex Decomposition}
\label{sec:map}
{In this section, we study the geometric properties of the feasible set $\mathcal{P}(P_{X,Y})$. We argue that there are two main challenges in characterizing $\mathcal{P}(P_{X,Y})$: i) the set is not convex and cannot be characterized using its extreme points. We overcome this challenge by decomposing the set into subsets which have desired convexity properties; and ii) evaluating $\mathcal{P}(P_{X,Y})$ involves studying discrete-output functions. Fourier analysis of such functions is challenging since in contrast with real-valued functions, the discrete-output condition puts a complex restriction on the Fourier coefficients of the functions --- that their linear combinations should lie in a discrete set. This makes their analysis and development of constructive algorithms difficult. We overcome this by providing an extension of the problem to real-valued functions along with a randomized rounding technique which maps these real-valued functions to discrete-output functions in a correlation-preserving manner.
As a first step, the following lemma shows that $\mathcal{P}(P_{X,Y})$ is not convex for non-trivial NISS scenarios, i.e., for scenarios where $X\neq Y$.}
\begin{Lemma}
The set $\mathcal{P}(P_{X,Y})$ is convex if and only if $X=Y$. 
\end{Lemma}
\begin{proof}
 For all $P_{X,Y}\in \Delta_{\mathcal{X}\times \mathcal{Y}}$ and  $a,b\in \mathcal{U}\times \mathcal{V}$, the distribution $P_{a,b}(u,v)= \mathbbm{1}(u=a,v=b)$ is in $\mathcal{P}(P_{X,Y})$. To show this, for a fixed $d\in \mathbb{N}$, take $f(x^d)=a, g(y^d)=b, \forall x^d,y^d\in \mathcal{X}^d\times \mathcal{Y}^d$ to produce $P_{a,b}$ from $X^d, Y^d$. On the other hand, for any given $P_{U,V}\in \Delta_{\mathcal{U}\times \mathcal{V}}$ we have $P_{U,V}= \sum_{a,b\in \mathcal{U}\times \mathcal{V}} \alpha_{a,b}P_{a,b}$, where $\alpha_{a,b}=P_{U,V}(a,b)$. So, the convex hull of $\mathcal{P}(P_{X,Y})$ is equal to $\Delta_{\mathcal{U}\times \mathcal{V}}$. In particular, the convex hull of  $\mathcal{P}(P_{X,Y})$  includes the distributon $P_{U,V}$ for which $P(U=V)=1$ and $U$
\end{proof}

As mentioned above, the fact that  $\mathcal{P}(P_{X,Y})$ is not a convex set makes its characterization particularly challenging. In the following, we provide a decomposition of the feasible region, $\mathcal{P}(P_{X,Y})$, whose components have desirable convexity properties which facilitate their characterization. In particular, we provide a bijective mapping from $\mathcal{P}(P_{X,Y})$ to a star-convex\footnote{A set $\mathcal{A}$ is called star-convex if there exits a point $e\in \mathcal{A}$ such that for any other point $e'\in \mathcal{A}$, the line connecting $e$ and $e'$ is in $\mathcal{A}$.} decomposition --- a union of star-convex sets. To this end, we first partition the set $\mathcal{P}(P_{X,Y})$. We define \[\mathcal{P}(P_{X,Y},Q_U,Q_V)\triangleq \{Q'_{U,V}\in \mathcal{P}(P_{X,Y})| Q'_U=Q_U, Q'_V=Q_V\}.\]
The collection $\mathcal{P}(P_{X,Y},Q_U,Q_V), Q_U\in \Delta_{\mathcal{U}}, Q_V\in \Delta_{\mathcal{V}}$ partitions the set
$\mathcal{P}(P_{X,Y})$. The following lemma provides some of the properties of associated sequences of functions of distributions in  $\mathcal{P}(P_{X,Y},Q_U,Q_V)$. These properties are used in constructing the aforementioned bijective mapping.
\begin{Lemma}
Given $Q_U\!\in\! \Delta_{\mathcal{U}}, Q_V\!\in \!\Delta_{\mathcal{V}}$ and $P_{X,Y}\!\in\! \Delta_{\mathcal{X}\times \mathcal{Y}}$, let $Q_{U,V}\!\in\! \mathcal{P}(P_{X,Y},Q_U,Q_V)$ and  $(f_d,g_d)_{d\in \mathbb{N}}$ be an associated sequence of functions of $Q_{U,V}$ defined in Definition \ref{def:1}. Define $f_{d,u}(\cdot)\triangleq\rchi(f_d(\cdot)= u), g_{d,v}(\cdot)\triangleq\rchi(g_{d}(\cdot)= v), u,v \in \mathcal{U}\times \mathcal{V}, d\in \mathbb{N}$.
Then, for all $u\in \mathcal{U}$ and $v\in \mathcal{V}$:
\\\textbf{i)} $\lim_{d\to \infty}\mathbb{E}(f_{d,u}(X^d))= 2Q_U(u)-1$.
\\\textbf{ii)}  $\lim_{d\to \infty}\mathbb{E}(g_{d,v}(Y^d))= 2Q_V(v)-1$.
\\\textbf{iii)} $\lim_{d\to \infty}\mathbb{E}(f_{d,u}(X^d)g_{d,v}(Y^d))$ exists.
\\Furthermore, for all $x^d,y^d\in \{-1,1\}$ and $d\in \mathbb{N}$:
\\\textbf{iv)} $f^2_{d,u}(x^d)\leq 1$ and $g^2_{d,v}(y^d)\leq 1$ for all $u\in \mathcal{U}, v\in  \mathcal{V}$.
\\\textbf{v)} $\sum_{u\in \mathcal{A}} f_{d,u}(x^d) \leq 2-|\mathcal{A}|$, for all $\mathcal{A} \subseteq \mathcal{U}$.
\\\textbf{vi)} $\sum_{v\in \mathcal{B}} g_{d,v}(y^d) \leq 2-|\mathcal{B}|$, for all $\mathcal{B} \subseteq \mathcal{V}$.
\label{prop:2}
\end{Lemma}
\begin{proof}
 Item i) follows by Equation \eqref{eq:2} for $g_v(y^d)=1, y^d\in \{-1,1\}^d$ and Definition \ref{def:1}. Item iii) follows by definition of the associated sequence of functions. Item iv) follows from the fact that $f_u, g_v\in \{-1,1\}, \forall u,v$. Item v) follows from the fact that for any given input $x^d\in \mathcal{X}^d$, only one of the functions $f_{d,u}(x^d), u\in \mathcal{U}$ outputs 1 and the rest output -1.
Items ii) and vi) follow by a similar argument as Item i) and v), respectively. 
\end{proof}
Given distributions $Q_U$ and $Q_V$, we denote the set of all sequences of associated functions of distributions in $\mathcal{P}(P_{X,Y},Q_U,Q_V)$ by $\widehat{\mathcal{F}}_{X,Y}(Q_U,Q_V)$. Note that the sequence $ (f_d,g_d)_{d\in \mathbb{N}}$ is in $\widehat{\mathcal{F}}_{X,Y}(Q_U,Q_V)$ if and only if it consists of discrete-output pairs of functions satisfying  conditions  i)-vi) in Lemma \ref{prop:2}. We wish to use the Fourier expansion methods described in Section \ref{sec:Fourier} to construct sequences of pairs of functions in $\widehat{\mathcal{F}}_{X,Y}(Q_U,Q_V)$ and evaluate their joint distribution. However, the fact that the functions in $\widehat{\mathcal{F}}_{X,Y}(Q_U,Q_V)$ are discrete-output enforces a restrictive structure on the Fourier coefficients which makes their analysis difficult. To overcome this challenge, we consider the extension $\mathcal{F}_{X,Y}(Q_U,Q_V)$ of $\widehat{\mathcal{F}}_{X,Y}(Q_U,Q_V)$ to real-valued functions. That is, we define the
sequences of real-valued function pairs $(f_d(\cdot), g_d(\cdot)), d\in \mathbb{N}$ satisfying conditions  i)-vi) in Lemma \ref{prop:2} by $\mathcal{F}_{X,Y}(Q_U,Q_V)$. 

Let $\mathcal{U}_\phi\triangleq\mathcal{U}-\{0\}$ and $\mathcal{V}_\phi\triangleq\mathcal{V}-\{0\}$. Define the set $\mathcal{E}(P_{X,Y},Q_U,Q_V))$ as follows:
\begin{align*}
 &   \mathcal{E}(P_{X,Y},Q_U,Q_V))\triangleq \bigg\{(e_{u,v}: u\in \mathcal{U}_\phi, v \in \mathcal{V}_\phi)\big|
 \\&\exists (f_d,g_d)_{d\in \mathbb{N}}\in   \mathcal{F}_{X,Y}(Q_U,Q_V): e_{u,v}=\lim_{d\to \infty} \mathbb{E}(f_{d,u}(X^d)g_{d,v}(Y^d))\bigg\}.
\end{align*}
We argue that there is a bijection between $\mathcal{P}(P_{X,Y},Q_U,Q_V)$ and $ \mathcal{E}(P_{X,Y},Q_U,Q_V)$. Furthermore, we show in Lemma \ref{prop:4} that $\mathcal{E}(P_{X,Y},Q_U,Q_V))$ is star-convex. This is used in Theorem \ref{th:1} to provide inner and outer bounds on $\mathcal{P}(P_{X,Y},Q_U,Q_V)$.
As an intermediary step, we introduce a randomized rounding technique which allows us to construct discrete-output pairs of functions from real-valued pairs of functions while maintaining some of their desired statistical properties such as expected value and Pearson correlation. 

\textbf{Correlation-Preserving Randomized Rounding Technique:} Let $(f_d(\cdot),g_d(\cdot))_{d\in \mathbb{N}}\in \mathcal{F}(Q_U,Q_V)$ be real-valued functions, and
define the mutually independent binary variables $E_{d,u}(x^d), d\in \mathbb{N}, u\in \mathcal{U}_\phi, x^d\in \{-1,1\}^d$ such that $P(E_{d,u}(x^d)=1)= \frac{1+f_{d,u}(x^d)}{2}$ and $P(E_{d,u}(x^d)=-1)= \frac{1-f_{d,u}(x^d)}{2}$. Similarly, define the variables $F_{d,v}\in \mathcal{V}_{\phi}$ with respect to $g_{d,v}(\cdot)$. Define the following sequence of binary-output (stochastic) functions:
\begin{align}
\label{eq:3}
    &\hat{f}_{d,u}(x^d)=\begin{cases}
    -1 \qquad & \text{if }\! \exists u'\!<\! u\!: \!E_{d,u'}(x^d)\!=\!1 \text{ or } E_{d,u}(x^d)\!=\!-1\\
    1& \text{ otherwise }
    \end{cases},
    \end{align}
    \begin{align}
&\hat{g}_{d,v}(y^d)=\begin{cases}
    -1 \qquad & \text{if }\! \exists v'\!<\! v\!:\! E_{d,v'}(y^d)\!=\!1 \text{ or } E_{d,v}(y^d)\!=\!-1\\
    1& \text{ otherwise}
    \end{cases},
    \label{eq:4}
\end{align}
for all $x^d,y^d \in \{-1,1\}^d, d\in \mathbb{N}$. It can be noted that this is an extension of the randomized rounding method introduced in \cite{ghazi2016decidability} for binary-output NISS. It is straightforward to verify that this is a \textit{correlation-preserving} procedure. That is, for all $d\in \mathbb{N}, u\in \mathcal{U}, v\in \mathcal{V}$ we have:
\begin{align*}
&\mathbb{E}(\hat{f}_{d,u}(X^d))= \mathbb{E}({f}_{d,u}(X^d)),\qquad\mathbb{E}(\hat{g}_{d,u}(Y^d))= \mathbb{E}({g}_{d,u}(Y^d))\\
&\mathbb{E}(\hat{f}_{d,u}(X^d)\hat{g}_{d,v}(Y^d))= \mathbb{E}({f}_{d,u}(X^d)g_{d,v}(Y^d)),
\end{align*}
Define
\begin{align*}
    &\hat{f}_{d}(\cdot)=
    \begin{cases}
   u\qquad &\text{ if }\exists!u\in \mathcal{U}_{\phi}:\hat{f}_{d,u}(\cdot)=1
   \\0 & \text{otherwise}
    \end{cases}
    \\&
    \hat{g}_{d}(\cdot)=
    \begin{cases}
    v\qquad& \text{ if }\exists!v\in \mathcal{V}_{\phi}:\hat{g}_{d,v}(\cdot)=1
   \\0 & \text{otherwise}
    \end{cases}
\end{align*} 
The functions $\hat{f}_d(\cdot)$ and $\hat{g}_d(\cdot)$ are well-defined  since by construction for each input $x^d$ and $y^d$ there is at most one $u$ and one $v$ such that $\hat{f}_{d,u}(x^d)=1$ and $\hat{g}_{d,v}(y^d)=1$, respectively. Furthermore, $(\hat{f}_d, \hat{g}_d)_{d\in \mathbb{N}} \in \widehat{\mathcal{F}}_{X,Y}(Q_U,Q_V)$.  Define $U_d=\hat{f}_d(X^d)$ and $V_d=\hat{g}_d(Y^d)$ for $d\in \mathbb{N}$ and let $Q_{d,U,V}$ be their joint distribution. From item vi) in Lemma \ref{prop:2}, $Q_{U,V}(u,v)\triangleq \lim_{d\to \infty} Q_{d,U,V}$ exists. 
\begin{Definition}[\textbf{Distribution Generated by Sequence of Functions}]
     Given $Q_U\in \Delta_{\mathcal{U}}$ and $Q_V\in \Delta_{\mathcal{U}}$, let $(f_d,g_d)_{d\in \mathbb{N}}\in \mathcal{F}(Q_U,Q_V)$. The distribution $Q_{U,V}$ generated through the randomized rounding procedure described in Equations \eqref{eq:3}-\eqref{eq:4} is called the distribution generated by   $(f_d,g_d)_{d\in \mathbb{N}}$.   
\end{Definition}

\begin{Lemma}
Given distributions $Q_U\in \Delta_{\mathcal{U}}$ and $Q_V\in \Delta_{\mathcal{V}}$ and $P_{X,Y}\in \Delta_{\mathcal{X}\times \mathcal{Y}}$, there exists a bijective mapping between $\mathcal{P}(P_{X,Y},Q_U,Q_V)$ and $ \mathcal{E}(P_{X,Y},Q_U,Q_V)$. 
\label{prop:3}
\end{Lemma}
\begin{proof}
Please refer to Appendix \ref{App:Prop:3}.
\end{proof}

The proof of Lemma \ref{prop:3} explicitly provides a bijective mapping from $\mathcal{P}(P_{X,Y},Q_U,Q_V)$ to $ \mathcal{E}(P_{X,Y},Q_U,Q_V)$. It follows that to characterize  $\mathcal{P}(P_{X,Y},Q_U,Q_V)$  it suffices to characterize $ \mathcal{E}(P_{X,Y},Q_U,Q_V)$. The following lemma shows that  $ \mathcal{E}(P_{X,Y},Q_U,Q_V)$ is star-convex.

\begin{Lemma}
Let $Q_U\in \Delta_{\mathcal{U}}$ and $Q_V\in \Delta_{\mathcal{V}}$ and $P_{X,Y}\in \Delta_{\mathcal{X}\times \mathcal{Y}}$. The set $\mathcal{E}(P_{X,Y},Q_U,Q_V)$ is a star-convex set. 
\label{prop:4}
\end{Lemma}
\begin{proof}
 Assume that $\mathcal{E}(P_{X,Y},Q_U,Q_V)\neq \phi$. 
 Let $\mathbf{e}^0= (e^0_{u,v})_{u\in \mathcal{U}_\phi, v\in \mathcal{V}_\phi}$, where $e^0_{u,v}= 2Q_U(u)Q_V(v)-1, \forall u,v\in \mathcal{U}_\phi\times \mathcal{V}_\phi$. The point  $\mathbf{e}^0$ is in $\mathcal{E}(P_{X,Y},Q_U,Q_V)$ as explained in the following. Let $f_d:\{-1,1\}^d\to \mathcal{U}$ be a sequence of functions such that for $U_d=f_d(X^n)$ its distribution $Q_{d,U}$ converges to $Q_U$ in variational distance as $d\to\infty$. Such a sequence exists since  $\mathcal{E}(P_{X,Y},Q_U,Q_V)\neq \phi$.
 Similarly, let  $g_d:\{-1,1\}^d\to \mathcal{V}$ be a sequence of functions such that for $V_d=g_d(Y^n)$ its distribution $Q_{d,V}$ converges to $Q_V$ in variational distance as $d\to\infty$. Define $f'_{2d}(x^{2d})\triangleq f_d(x^d), x^{2d}\in \{-1,1\}^d$, so that $f'_{2d}$ depends only on the first half of $x^{2d}$. Similarly, define $g'_{2d}(y^{2d})\triangleq g_d(y_{d+1}^{2d}), y^{2d}\in \{-1,1\}^d$, so that $g'_{2d}$ depends only on the second half of $y^{2d}$. Then, $\lim_{d\to\infty} \mathbb{E}(f'_{2d,u}(X^d)g'_{2d,v}(Y^d))=2Q_U(u)Q_V(v)-1$\footnote{One can define the odd-indexed functions $f'_{2d+1}$ and $g'_{2d+1}$ to be equal to the even-indexed $f'_{2d}$ and $g'_{2d}$, respectively, so that the limit is well-defined.} and $(f'_{d},g'_{d})_{d\in \mathbb{N}}\in \mathcal{F}(Q_U,Q_V)$, so $\mathbf{e}^0\in\mathcal{E}(P_{X,Y},Q_U,Q_V)$. 

 We will show that for any point $\mathbf{e}$ in  $\mathcal{E}(P_{X,Y},Q_U,Q_V)$, the line connecting $\mathbf{e}$ to $\mathbf{e}_0$ lies inside  $\mathcal{E}(P_{X,Y},Q_U,Q_V)$, and hence the set is star-convex. To see this, let $(f_d, g_d)$ be the associated sequence of functions for $\mathbf{e}$. Define $f'_{2d}(x^{2d})\triangleq f_d(x^d), x^{2d}\in \{-1,1\}^d$ as in the previous step, and define
 define the following stochastic functions:
 \begin{align*}
     g^{(p)}_{2d}(y^{2d}) = \begin{cases}
     g_d(y_{d+1}^{2d}) \quad & \text{with probability } p,\\
     g_d(y^d) & \text{with probability } 1-p,
     \end{cases}, p\in [0,1].
 \end{align*}
 Then, $\lim_{d\to\infty} \mathbb{E}(f'_{2d,u}(X^d)g^{(p)}_{2d,v}(Y^d))=
 pe^0_{u,v}+(1-p)e_{u,v}, u,v\in \mathcal{U}_{\phi}\times\mathcal{V}_{\phi}$. So, the point $p\mathbf{e}^0+(1-p)\mathbf{e}$ is in $\mathcal{E}(P_{X,Y},Q_U,Q_V)$. This completes the proof.
\end{proof}

The following Lemma leverages the star-convexity of $\mathcal{E}(P_{X,Y},Q_U,Q_V)$ to provide inner-bounds, i.e. to characterize subsets of $\mathcal{E}(P_{X,Y},Q_U,Q_V)$.
\begin{Lemma}
\label{prop:5}
Given distributions $Q_U\in \Delta_{\mathcal{U}}$ and $Q_V\in \Delta_\mathcal{V}$, let $(f_d(\cdot),g_d(\cdot))_{d\in \mathbb{N}}\in \mathcal{F}(Q_U,Q_V)$. Let $Q_{U,V}$ be the distribution generated by $(f_d(\cdot),g_d(\cdot))_{d\in \mathbb{N}}$. 
For any $\mathcal{A}\subseteq \mathcal{U}_{\phi}$ and $\mathcal{B}\subseteq \mathcal{V}_{\phi}$ define $\mathbf{s}_{\mathcal{A},\mathcal{B}}= (s_{\mathcal{A},\mathcal{B},u,v}: u\in \mathcal{U}_{\phi}, v\in \mathcal{V}_{\phi})$, where
\begin{align}
\label{eq:6}
 s_{\mathcal{A},\mathcal{B},u,v}= 
 \begin{cases}
 2Q_{U,V}(u,v)-1 \qquad & \text{ if } u\in \mathcal{A}_{\phi}, v\in \mathcal{B}_{\phi}\\
2Q_U(u)Q_V(v)-1 & \text{otherwise.}
 \end{cases}
\end{align}
Then, the convex polytope $\mathcal{S}_{Q_{U,V}}$ generated by the set of vertices $\{\mathbf{s}_{\mathcal{A},\mathcal{B}},\mathcal{A}\subseteq \mathcal{U}_{\phi}, \mathcal{B}\subseteq \mathcal{V}_{\phi}\}$ is a subset of $\mathcal{E}(P_{X,Y},Q_U,Q_V)$. 
\end{Lemma}
The proof follows by similar arguments as the randomization approach used in the proof of Lemma \ref{prop:4} and is omitted for brevity.

\begin{Definition}[\textbf{Supporting Vectors}]
\label{def:3}
For an NISS problem characterized by $(P_{X,Y},\mathcal{U}, \mathcal{V})$, consider distributions $Q_U\in \Delta_{\mathcal{U}}, Q_V\in \Delta_{\mathcal{V}}$, and $\pmb{\lambda}=(\lambda_{u,v}, u\in \mathcal{U}_\phi, \mathcal{V}_\phi)$, where $\lambda_{u,v}\geq 0, \forall u,v$ and $\sum_{u,v}\lambda_{u,v}=1$.
Let $\mathbf{e}^{\pmb{\lambda,+}}\triangleq(e^{\pmb{\lambda},+}_{u,v},u\in \mathcal{U}_\phi,v\in \mathcal{V}_\phi)$ and $\mathbf{e}^{\pmb{\lambda,-}}\triangleq(e^{\pmb{\lambda},-}_{u,v},u\in \mathcal{U}_\phi, v\in \mathcal{V}_\phi)$ be the extreme points of $\mathcal{E}(P_{X,Y},Q_U,Q_V)$ corresponding to the hyperplane characterized by the vector $\pmb{\lambda}$. That is:
\begin{align*}
&\textbf{e}^{\pmb{\lambda},-}\triangleq\argmin_{\mathbf{e}\in \mathcal{E}(P_{X,Y},Q_U,Q_V)} \sum_{u\in \mathcal{U}_{\phi},v\in \mathcal{V}_{\phi}}  \lambda_{u,v} e_{u,v},
\\&\textbf{e}^{\pmb{\lambda},+}\triangleq\argmax_{\mathbf{e}\in \mathcal{E}(P_{X,Y},Q_U,Q_V)} \sum_{u\in \mathcal{U}_{\phi},v\in \mathcal{V}_{\phi}}  \lambda_{u,v} e_{u,v}.
\end{align*}
The set of supporting vectors of $\mathcal{E}(P_{X,Y},Q_U,Q_V)$ is denoted by $\mathbf{E}\triangleq\cup_{\pmb{\lambda}} \{\textbf{e}^{\pmb{\lambda},+}, \textbf{e}^{\pmb{\lambda},-}\}$. We define $Q_{U,V}^{\pmb{\lambda},+}\triangleq \Xi^{-1}(\textbf{e}^{\pmb{\lambda},+})$ and 
$Q_{U,V}^{\pmb{\lambda},-}\triangleq \Xi^{-1}(\textbf{e}^{\pmb{\lambda},-})$, where $\Xi$ is defined in the proof of Lemma \ref{prop:3}.
\end{Definition}

The following theorem provides inner and outer bounds on the set of feasible distributions in the NISS problem. The outer bounds are further simplified into a closed-form computable expression in Theorem \ref{th:2}.
\begin{Theorem}
\label{th:1}
Given an NISS problem characterized by $(P_{X,Y},\mathcal{U}, \mathcal{V})$, the following provides inner and outer bounds on its set of feasible distributions:
\begin{align*}
    & \Xi\left(\bigcup_{Q_U\in \Delta_{\mathcal{U}},Q_V\in \Delta_{\mathcal{V}} }\bigcup_{\pmb{\lambda}}{(\mathcal{S}_{Q^{\pmb{\lambda},+}_{U,V}}\cup \mathcal{S}_{Q^{\pmb{\lambda},-}_{U,V}})}\right)
    \subseteq \mathcal{P}(P_{X,Y})\subseteq
    \\&\qquad \qquad  \Xi\left(Conv\left(\bigcup_{Q_U\in \Delta_{\mathcal{U}},Q_V\in \Delta_{\mathcal{V}} }\bigcup_{\pmb{\lambda}}{(\mathcal{S}_{Q^{\pmb{\lambda},+}_{U,V}}\cup \mathcal{S}_{Q^{\pmb{\lambda},-}_{U,V}})}\right)\right),
\end{align*}
where $Q^{\pmb{\lambda},+}_{U,V}$ and $Q^{\pmb{\lambda},+}_{U,V}$ are defined in Definition \ref{def:3} and $\mathcal{S}_{Q_{U,V}}$ is defined in Lemma \ref{prop:5}.\footnote{Note that $Q^{\pmb{\lambda},+}_{U,V}$ and $Q^{\pmb{\lambda},-}_{U,V}$ are defined as a function of $\pmb{\lambda}$, $Q_U$ and $Q_V$, but the dependence on $Q_U$ and $Q_V$ is not made explicit to simplify the notation.}
\end{Theorem}
The proof of the inner-bound follows from Lemma \ref{prop:5} and Definition \ref{def:3}. The proof of the outer-bound follows from Lemma \ref{prop:2} and the fact that any set is a subset of its convex hull. 

\section{A Computable Outer Bound}
Theorem \ref{th:1} implies that in order to derive outer bounds on the set of feasible distributions in the NISS problem, it suffices to characterize the supporting hyperplanes of $\mathcal{E}(P_{X,Y},Q_U,Q_V)$, since the convex hull can be written as the  polytope characterized by the intersection of the closed half-spaces created by the supporting hyperplanes (e.g. \cite{boyd2004convex}).
Let us define
\begin{align}
\label{eq:7}
&t^{\pmb{\lambda},-}=\min_{\mathbf{e}\in \mathcal{E}(P_{X,Y},Q_U,Q_V)} \sum_{u\in \mathcal{U}_{\phi},v\in \mathcal{V}_{\phi}}  \lambda_{u,v} e_{u,v},
\\&{t}^{\pmb{\lambda},+}=\max_{\mathbf{e}\in \mathcal{E}(P_{X,Y},Q_U,Q_V)} \sum_{u\in \mathcal{U}_{\phi},v\in \mathcal{V}_{\phi}}  \lambda_{u,v} e_{u,v}.
\label{eq:8}
\end{align}
The following theorem provides bounds on ${t}^{\pmb{\lambda},-},{t}^{\pmb{\lambda},+}$ which in turn provide inner and outer bounds on the set of feasible distributions.  
\begin{Theorem}
\label{th:2}
For the NISS problem characterized by $(P_{X,Y},\mathcal{U},\mathcal{V})$. The following holds:
\[
    \theta^-\leq t^{\pmb{\lambda},-}\leq  t^{\pmb{\lambda},+}\leq \theta^+,\]
where 
\begin{align*}
    &\theta^+\triangleq \theta_{\phi} +2\rho \theta_\rho +\frac{1}{2}\rho^2(-\theta_{\rho}+\theta_{\rho^2,1}+\theta_{\rho^2,3})
    \\&\theta^-\triangleq \theta_\phi-2\rho  \theta_{\rho}-
    \frac{1}{2}\rho^2\Big( \theta_{\rho^2,2}+\theta_{\rho^2,3})
\\
& \theta_\phi\triangleq \sum_{u\in \mathcal{U}_{\phi},v\in \mathcal{V}_{\phi}}  \lambda_{u,v}(2Q_U(u)-1)(2Q_V(v)-1)
\\&
\theta_\rho\triangleq  \sum_{u\in \mathcal{U}_\phi, v\in \mathcal{V}_\phi}\lambda_{u,v}\sqrt{Q_U(u)Q_V(v)}
\\&\theta_{\rho^2,1}\triangleq \sum_{u\in \mathcal{U}_{\phi},v\in \mathcal{V}_{\phi}}\!\!\!\!\! \! \lambda_{u,v}Q_U(u)(1-Q_V(v))\\
&\theta_{\rho^2,2}\triangleq \sum_{u\in \mathcal{U}_{\phi},v\in \mathcal{V}_{\phi}}\!\!\!\!\! \! \lambda_{u,v}Q_U(u)Q_V(v)\\
&\theta_{\rho^2,3}\triangleq\!\!\sqrt{\sum_{u\in \mathcal{U}_{\phi},v\in \mathcal{V}_{\phi}}\!\!\!\!\! \! \lambda_{u,v}Q_U(u)(1\!-\!Q_U(u))
\!\!\!\!\!\!\sum_{u\in \mathcal{U}_{\phi},v\in \mathcal{V}_{\phi}} \!\!\!\!\!\! \lambda_{u,v}Q_V(v)(1\!-\!Q_V(v))})
.
\end{align*}

\end{Theorem}
\begin{Remark}
Note that if $|\mathcal{U}|=|\mathcal{V}|=2$, then $\mathcal{U}_\phi$ and $\mathcal{V}_\phi$ have one element, and $\lambda_{1,1}=1$. This recovers the bound given in \cite[Theorem 1]{yu2022common} for binary output NISS. 
\end{Remark}
\begin{proof}
Please refer to Appendix \ref{App:Th:2}.
\end{proof}

\section{Conclusion}
The NISS  problem was considered.
Inner and outer bounds were obtained on the set of distributions $Q_{U,V}$ which can be produced given an input distribution $P_{X,Y}$. The derivation and proof techniques were based on discrete Fourier analysis along with a novel randomized rounding technique. The bounds are applicable in NISS scenarios where the output alphabets have arbitrary finite size. A future avenue of research is to extend the derivations to general ergodic distributed sources.
\begin{appendices}
    \section{Proof of Lemma \ref{prop:3}}
    \label{App:Prop:3}
    By the Cantor-Bernstein Theorem (e.g. \cite{hinkis2013proofs}), it suffices to show that there exist injective functions from each set to the other. To show the existence of an injective function from $\mathcal{P}(P_{X,Y},Q_U,Q_V)$ to $\mathcal{E}(P_{X,Y},Q_U,Q_V)$, let us take $Q_{U,V}\in \mathcal{P}(P_{X,Y},Q_U,Q_V)$, and let $(f_d,g_d),d\in \mathbb{N}$ be its associated sequence of functions. Then, by Lemma \ref{prop:2}, we have 
$(f_d,g_d)_{d\in \mathbb{N}}\in \mathcal{F}_{X,Y}(Q_U,Q_V)$. Furthermore, from Equation \eqref{eq:2}, for any $u,v\in \mathcal{U}_\phi\times \mathcal{V}_\phi$ we have:
\begin{align}
&\lim_{d\to \infty} \mathbb{E}(f_{d,u}(X^d)g_{d,v}(Y^d)) = 2 \lim_{d\to \infty}P(f_{d,u}(X^d)= g_{d,v}(Y^d))-1
=2Q_{U,V}(U=u, V= v)-1\triangleq q_{u,v}, 
\label{eq:5}
\end{align}
where we have used the fact that from Lemma \ref{prop:2}, $Q_{U,V}(u,v)= \lim_{d\to \infty} Q_{d,U,V}$ exists. 
We define the mapping $\Psi:\mathcal{P}(P_{X,Y},Q_U,Q_V)\to \mathcal{E}(P_{X,Y},Q_U,Q_V)$ such that $\Psi: Q_{U,V}\mapsto (q_{u,v})_{u\in \mathcal{U}_\phi,v\in \mathcal{V}_\phi}$. The mapping $\Psi$ is injective. To see this assume that $\Psi(Q_{U,V})=\Psi(Q'_{U,V})$ for $Q_{U,V},Q'_{U,V}\in \mathcal{P}(P_{X,Y},Q_U,Q_V)$. Then, $Q_{U,V}$ and $Q'_{U,V}$ must have the same marginals $Q_U$ and $Q_V$ and should both satisfy Equation \eqref{eq:5}. These two conditions provide a system of $|\mathcal{U}||\mathcal{V}|$ linearly independent equalities. Since $Q_{U,V}$ is a solution for this system, it is the unique solution. Hence, we must have $Q_{U,V}=Q'_{U,V}$. 

Next, we construct an injective mapping in the reverse direction.
Let $\Xi:\mathcal{E}(P_{X,Y},Q_U,Q_V) \to  \mathcal{P}(P_{X,Y},Q_U,Q_V)$ be such that $\Xi: (e_{u,v})_{u\in \mathcal{U}_\phi, v\in \mathcal{V}_\phi}\mapsto Q_{U,V}$ with $Q_{U,V}$ being the distribution generated by the associated sequence of function of $(e_{u,v})_{u\in \mathcal{U}_\phi, v\in \mathcal{V}_\phi}$. Then, it is straightforward to see that $\Xi$ is an injective mapping. This completes the proof.
\section{Proof of Theorem \ref{th:2}}
\label{App:Th:2}
Note that 
\begin{align*}
 &\sum_{u\in \mathcal{U}_{\phi},v\in \mathcal{V}_{\phi}}  \lambda_{u,v} \mathbb{E}(f_{d,u}(X^d)g_{d,v}(Y^d))
=\sum_{u\in \mathcal{U}_{\phi},v\in \mathcal{V}_{\phi}} \sum_{\mathcal{S}\subseteq [d]} \lambda_{u,v} f_{d,u,\mathcal{S}}g_{d,v,\mathcal{S}}\rho^{|\mathcal{S}|}
 \\&=\sum_{u\in \mathcal{U}_{\phi},v\in \mathcal{V}_{\phi}}  \lambda_{u,v} f_{d,u,\phi}g_{d,v,\phi}+
 \sum_{u\in \mathcal{U}_{\phi},v\in \mathcal{V}_{\phi}} \sum_{\mathcal{S}:|\mathcal{S}|\geq 1} \lambda_{u,v} f_{d,u,\mathcal{S}}g_{d,v,\mathcal{S}}\rho^{|\mathcal{S}|}
\end{align*}

Note that $f_{d,u,\phi}= 2Q_U(u)-1$ and $g_{d,v,\phi}= 2Q_V(v)-1$. So, 
\begin{align*}
   & \sum_{u\in \mathcal{U}_{\phi},v\in \mathcal{V}_{\phi}}  \lambda_{u,v} f_{d,u,\phi}g_{d,v,\phi}
    =  \sum_{u\in \mathcal{U}_{\phi},v\in \mathcal{V}_{\phi}}  \lambda_{u,v}(2Q_U(u)-1)(2Q_V(v)-1).
\end{align*}
To bound the second term, we define the sets $\mathcal{A}_{u,v}\triangleq\{\mathcal{S}\subseteq [d]: |\mathcal{S}|\geq 2, f_{d,u,\mathcal{S}}g_{d,v,\mathcal{S}}\geq 0\}$ and  $\mathcal{B}_{u,v}\triangleq\{\mathcal{S}\subseteq [d]: |\mathcal{S}|\geq 2, f_{d,u,\mathcal{S}}g_{d,v,\mathcal{S}}< 0\}$. Furthermore, we define:
\begin{align*}
&\tau^+_{u,v}\triangleq \sum_{\mathcal{S}\in \mathcal{A}} \lambda_{u,v} f_{d,u,\mathcal{S}}g_{d,v,\mathcal{S}}, \quad u\in \mathcal{U}_\phi, v\in \mathcal{V}_{\phi}
\\&\tau^-_{u,v}\triangleq \sum_{\mathcal{S}\in \mathcal{B}} \lambda_{u,v} f_{d,u,\mathcal{S}}g_{d,v,\mathcal{S}},\quad u\in \mathcal{U}_\phi, v\in \mathcal{V}_{\phi}.
\end{align*}
Note that: 
\begin{align*}
  &\sum_{u\in \mathcal{U}_{\phi},v\in \mathcal{V}_{\phi}}  \lambda_{u,v} (\rho \sum_{\mathcal{S}:|\mathcal{S}|= 1} f_{d,u,\mathcal{S}}g_{d,v,\mathcal{S}}+ \rho^2 \tau^-_{u,v})\leq 
\sum_{u\in \mathcal{U}_{\phi},v\in \mathcal{V}_{\phi}} \sum_{\mathcal{S}:|\mathcal{S}|\geq 1} \lambda_{u,v} f_{d,u,\mathcal{S}}g_{d,v,\mathcal{S}}\rho^{|\mathcal{S}|}\leq 
 \\&
  \sum_{u\in \mathcal{U}_{\phi},v\in \mathcal{V}_{\phi}}  \lambda_{u,v} (\rho \sum_{\mathcal{S}:|\mathcal{S}|= 1} f_{d,u,\mathcal{S}}g_{d,v,\mathcal{S}}+ \rho^2 \tau^+_{u,v})
\end{align*}
Using \cite[Equation (48)]{yu2022common}, for all $ u\in \mathcal{U}, v\in \mathcal{V}$, we get:
\begin{align}
|\!\!\sum_{\mathcal{S}:|\mathcal{S}|= 1} \!\lambda_{u,v} f_{d,u,\mathcal{S}}g_{d,v,\mathcal{S}}|\leq 
 2 \lambda_{u,v} \sqrt{Q_U(u)Q_V(v)},
 \label{eq:C}
\end{align}
{Therefore,}
\begin{align*}
&\sum_{u\in \mathcal{U}_{\phi},v\in \mathcal{V}_{\phi}} \sum_{\mathcal{S}:|\mathcal{S}|= 1} \lambda_{u,v} f_{d,u,\mathcal{S}}g_{d,v,\mathcal{S}}\leq 
 2\sum_{u\in \mathcal{U}_{\phi},v\in \mathcal{V}_{\phi}} \lambda_{u,v}\sqrt{Q_U(u)Q_V(v)},
\end{align*}
We need to find upper bounds and lower bounds on $\tau^+_{u,v}$ and $ \tau^-_{u,v}$, respectively. Using the fact that $|\rho|\leq 1$,
we have:
\begin{align*}
&\sum_{u\in \mathcal{U}_{\phi},v\in \mathcal{V}_{\phi}} (\sum_{\mathcal{S}:|\mathcal{S}|= 1} \lambda_{u,v} f_{d,u,\mathcal{S}}g_{d,v,\mathcal{S}}+\tau^+_{u,v}-\tau^-_{u,v}) 
\leq \sum_{u\in \mathcal{U}_{\phi},v\in \mathcal{V}_{\phi}} \sum_{\mathcal{S}:|\mathcal{S}|\geq  1}  |\sqrt{\lambda_{u,v}}f_{d,u,\mathcal{S}}||\sqrt{\lambda_{u,v}}g_{d,v,\mathcal{S}}|
\end{align*}
\begin{align*}
&\leq 
\sqrt{\sum_{u\in \mathcal{U}_{\phi},v\in \mathcal{V}_{\phi}} \sum_{\mathcal{S}:|\mathcal{S}|\geq  1} \lambda_{u,v}f^2_{d,u,\mathcal{S}}\sum_{u\in \mathcal{U}_{\phi},v\in \mathcal{V}_{\phi}} \sum_{\mathcal{S}:|\mathcal{S}|\geq  1} \lambda_{u,v}g^2_{d,v,\mathcal{S}}}
\leq 
\sqrt{\sum_{u\in \mathcal{U}_{\phi},v\in \mathcal{V}_{\phi}}  \lambda_{u,v}(1-f^2_{d,u,\phi})\sum_{u\in \mathcal{U}_{\phi},v\in \mathcal{V}_{\phi}}  \lambda_{u,v}(1-g^2_{d,v,\phi})}
\\\numberthis & = 
\sqrt{\sum_{u\in \mathcal{U}_{\phi},v\in \mathcal{V}_{\phi}}\!\!\!\!\! \! \lambda_{u,v}Q_U(u)(1-Q_U(u))
\!\!\!\!\!\sum_{u\in \mathcal{U}_{\phi},v\in \mathcal{V}_{\phi}} \!\!\!\!\!\! \lambda_{u,v}Q_V(v)(1-Q_V(v))}.
\label{eq:13}
\end{align*}
Furthermore, using the non-negativity of probabilities,(e.g.,\cite[Equations (23)-(26)]{yu2022common}) we have:
\begin{align}
&\nonumber
- \lambda_{u,v} Q_U(u)Q_V(v)
\leq 
 \sum_{\mathcal{S}:|\mathcal{S}|= 1} \lambda_{u,v} f_{d,u,\mathcal{S}}g_{d,v,\mathcal{S}}+\tau^+_{u,v}+\tau^-_{u,v}\leq  
 \\&
 \lambda_{u,v}  Q_U(u)(1-Q_V(v)), \quad u\in \mathcal{U}_\phi, v\in \mathcal{V}_\phi.
 \label{eq:14}
\end{align}
Consequently, from Equations \eqref{eq:13} and \eqref{eq:14}, we have:
\begin{align}
\nonumber
&   \sum_{u\in \mathcal{U}_{\phi},v\in \mathcal{V}_{\phi}} \tau^+_{u,v} \leq -\sum_{\mathcal{S}:|\mathcal{S}|= 1} \lambda_{u,v} f_{d,u,\mathcal{S}}g_{d,v,\mathcal{S}}\\
&
+ \frac{1}{2}\sum_{u\in \mathcal{U}_{\phi},v\in \mathcal{V}_{\phi}}\!\!\!\!\! \! \lambda_{u,v}Q_U(u)(1-Q_V(v))
   \label{eq:tplus}+\frac{1}{2}\sqrt{\sum_{u\in \mathcal{U}_{\phi},v\in \mathcal{V}_{\phi}}\!\!\!\!\! \! \lambda_{u,v}Q_U(u)(1-Q_U(u))
\!\!\!\!\!\sum_{u\in \mathcal{U}_{\phi},v\in \mathcal{V}_{\phi}} \!\!\!\!\!\! \lambda_{u,v}Q_V(v)(1-Q_V(v))}
\end{align} 
and 
\begin{align}
&   \sum_{u\in \mathcal{U}_{\phi},v\in \mathcal{V}_{\phi}} \tau^-_{u,v}\geq 
-\frac{1}{2}\Big(\sum_{u\in \mathcal{U}_{\phi},v\in \mathcal{V}_{\phi}}\!\!\!\!\! \! \lambda_{u,v}Q_U(u)Q_V(v)+
\label{eq:tminus}
\sqrt{\sum_{u\in \mathcal{U}_{\phi},v\in \mathcal{V}_{\phi}}\!\!\!\!\! \! \lambda_{u,v}Q_U(u)(1-Q_U(u))
\!\!\!\!\!\sum_{u\in \mathcal{U}_{\phi},v\in \mathcal{V}_{\phi}} \!\!\!\!\!\! \lambda_{u,v}Q_V(v)(1-Q_V(v))}\Big).
\end{align}
The proof follows by noting that:
\begin{align*}
    &t^{\pmb{\lambda},+}\leq \sum_{u\in \mathcal{U}_{\phi},v\in \mathcal{V}_{\phi}}  \lambda_{u,v}(2Q_U(u)-1)(2Q_V(v)-1)
+ \sum_{u\in \mathcal{U}_{\phi},v\in \mathcal{V}_{\phi}}  \lambda_{u,v} (\rho \sum_{\mathcal{S}:|\mathcal{S}|= 1} f_{d,u,\mathcal{S}}g_{d,v,\mathcal{S}}+ \rho^2 \tau^+_{u,v})
    \\&\leq  \theta_{\phi} +2\rho \theta_\rho +\frac{1}{2}\rho^2(-\theta_{\rho}+\theta_{\rho^2,1}+\theta_{\rho^2,3}),
\end{align*}
where we have used Equations \eqref{eq:C} and \eqref{eq:tplus} in the last inequality. The bound on $t^{\pmb{\lambda},-}$ is derived similarly.
\qed
\end{appendices}

\bibliographystyle{unsrt}

\end{document}